\documentclass[a4paper,UKenglish,cleveref, autoref, thm-restate]{lipics-v2021}

\title{A polynomial delay algorithm generating all potential maximal cliques in triconnected planar graphs} 

\titlerunning{Potential maximal cliques in triconnected planar graphs} 

\author{Alexander Grigoriev}
{Department of Data Analytics and Digiatlisation, Maastricht University, Maastricht, The Netherlands}
{a.grigoriev@maastrichtuniversity.nl}
{https://orcid.org/0000-0002-8391-235X}
{}

\author{Yasuaki Kobayashi}
{Hokkaido University, Sapporo, Japan}
{koba@ist.hokudai.ac.jp}
{https://orcid.org/0000-0003-3244-6915}
{JSPS KAKENHI Grant Numbers JP23K28034, JP24H00686 and JP24H00697}

\author{Hisao Tamaki}
{Meiji University, Kawasaki, Japan}
{hisao.tamaki@gmail.com}
{https://orcid.org/0000-0001-7566-8505}
{JSPS KAKENHI Grant Number JP24H00697}

\author{Tom C. van der Zanden}
{Department of Data Analytics and Digiatlisation, Maastricht University, Maastricht, The Netherlands}
{t.vanderzanden@maastrichtuniversity.nl}
{https://orcid.org/0000-0003-3080-3210}
{}

\authorrunning{A. Grigoriev, Y. Kobayashi, H. Tamaki and T. C. van der Zanden} 

\Copyright{Alexander Grigoriev, Yasuaki Kobayashi, Hisao Tamaki and Tom C. van der Zanden} 



\relatedversion{} 

\nolinenumbers 

\EventEditors{Akanksha Agrawal and Erik Jan van Leeuwen}
\EventNoEds{2}
\EventLongTitle{20th International Symposium on Parameterized and Exact Computation (IPEC 2025)}
\EventShortTitle{IPEC 2025}
\EventAcronym{IPEC}
\EventYear{2025}
\EventDate{September 17--19, 2025}
\EventLocation{Warsaw, Poland}
\EventLogo{}
\SeriesVolume{358}
\ArticleNo{21}
\usepackage{todonotes}
\usepackage{algorithm}
\usepackage{algpseudocode}
\usepackage{here}

\usepackage{xcolor}
\hypersetup{
	colorlinks=true,
	citecolor=blue,
	linkcolor=blue,
}

\newcommand{\calC}{{\mathcal{C}}}
\newcommand{\calT}{{\mathcal{T}}}
\newcommand{\calS}{{\mathcal{S}}}
\newcommand{\calP}{{\mathcal{P}}}
\newcommand{\calX}{{\mathcal{X}}}
\newcommand{\sphere}{\Sigma}
\newcommand{\tw}{{\mathop{tw}}}
\newcommand{\gen}{{\textsc{Gen}}}
\newcommand{\LineComment}[1]{\newline \(\triangleright\) #1}

\begin{document}

\maketitle
\begin{abstract}
    We develop a new characterization of potential maximal cliques of a triconnected planar graph and, using this characterization,
    give a polynomial delay algorithm generating all potential maximal cliques of a given triconnected planar graph.
    Combined with the dynamic programming algorithm due to Bouchitt{\'e} and Todinca, this algorithm leads to a 
    treewidth algorithm for general planar graphs that runs in time linear in the number of potential maximal cliques and
    polynomial in the number of vertices.
\end{abstract}

\newpage
\section{Introduction}
Let $G$ be a graph. A vertex set $X$ of $G$ is a \emph{potential maximal clique} (PMC) of $G$
if there is a minimal triangulation $H$ of $G$ such that $X$ is a maximal clique of $H$.
Let $\Pi(G)$ denote the set of all PMCs of $G$. Computing $\Pi(G)$ is
the first step in the treewidth algorithms of Bouchitt{\'e} and Todinca \cite{bouchitte2001treewidth}.
The second step is a dynamic programming algorithm that works on $\Pi(G)$ and
computes the treewidth $\tw(G)$ of $G$, together with a tree-decomposition of $G$ achieving the width,
in time $|\Pi(G)|n^{O(1)}$. The same approach works for various problems that can be formulated as
asking for an optimal triangulation with some criterion, including the minimum fill-in problem
\cite{bouchitte2001treewidth}. 
See \cite{FominTV15} and \cite{FominV10} 
for more applications of PMCs in this direction.

Naturally, the complexity of computing $\Pi(G)$ is of great interest. In a separate paper
\cite{bouchitte2002listing}, Bouchitt{\'e} and Todinca showed that $\Pi(G)$ can be computed
in time polynomial in the number of minimal separators (see Section~\ref{sec:prelim} for the definition) of $G$. 
Fomin and Villanger \cite{fomin2012treewidth} showed that $|\Pi(G)|$ is $O(1.7549^n)$ and gave an
algorithm to compute $\Pi(G)$ in time $O(1.7549^n)$.
(These bounds were later improved to $O(1.7346^n)$ in \cite{FominV10}.)
Although this bound on the running time matches
the combinatorial bound, an algorithm with running time $|\Pi(G)| n^{O(1)}$ is much more desirable in practice.
In fact, the success of recent practical algorithms for treewidth based
on the Bouchitt{\'e}-Todinca algorithm, \cite{tamaki2019positive} for example, relies on the fact that $|\Pi(G)|$ is vastly smaller than
this theoretical bound on many instances of interest in practice.

In this paper, we address this question if $\Pi(G)$ can be computed in time $|\Pi(G)| n^{O(1)}$, which is one of the
open questions in parameterized and exact computation listed in \cite{bodlaender2006open}.
We note that this question has been open for any natural graph class except those for which
$\Pi(G)$ is known to be computable in polynomial time. We answer this question in the affirmative
for the class of triconnected planar graphs. 
As a straightforward consequence, we obtain a new running time upper bound 
of $|\Pi(G)| n^{O(1)}$ on the treewidth computation, not only for triconnected planar graphs but
also for general planar graphs, since the treewidth of a planar graph can be computed by
working separately on its triconnected components.

This result is also interesting with regard to a broader question: how planarity helps in treewidth computation?
In the case of branchwidth, a graph parameter similar to treewidth, 
there is a clear answer to the similar question. Branchwidth is NP-hard to compute for general graphs but
polynomial time computable for planar graphs by the celebrated Ratcatcher algorithm due to Seymour and Thomas \cite{seymour1994call}. 
It has been a long open question whether an analogy
holds for treewidth: we do not know whether computing treewidth is NP-hard or polynomial time solvable on planar graphs.
In addition to the 1.5-approximation algorithm based on the Ratcatcher algorithm for branchwidth,
several approximation algorithms for treewidth are known that exploit planarity (\cite{KAMMER201660} for example, also see citations therein), but no non-trivial exact algorithm for planar treewidth is known.
Several fixed-parameter tractable exact treewidth algorithms for general graphs are known,
including a $2^{O(k^3)} n$ time algorithm \cite{10.1137/S0097539793251219} and
a $2^{O(k^2)} n^{O(1)}$ time algorithm \cite{korhonen2023improved}.
No improvement, however, on these results exploiting planarity is known.

Our algorithm is the first one that explicitly exploits planarity for computing exact treewidth
in a non-trivial manner.

Our algorithm not only computes $\Pi(G)$ for triconnected planar graph $G$ in
time $|\Pi(G)|n^{O(1)}$ but also runs with polynomial delay: it spends $n^{O(1)}$ time for
the generation of each element in the set. See Section~\ref{sec:poly_delay} for a more technical definition of
polynomial delay generation. This feature of our algorithm is appealing in the
field of combinatorial enumeration where designing an algorithm
with polynomial delay is one of the main research goals.
We note that most of the devices we use in our algorithm are still needed
even if we downgrade our requirement of polynomial delay 
to the total running time bound of
$|\Pi(G)|n^{O(1)}$.

Our generation algorithm is based on a new and simple characterization of PMCs of triconnected planar graphs.
Fomin, Todinca, and Villanger \cite{fomin2011exact} used a characterization of PMCs of general planar graphs in terms of
what they call \emph{$\theta$-structures}. A $\theta$-structure is, roughly speaking, a collection of three curves sharing their ends in the sphere of embedding such that each pair forms a \emph{noose} of the embedded graph, where a noose is a simple closed curve that
intersects the embedded graph only at vertices.
They used this characterization to design an algorithm for the maximum induced planar subgraph problem.
Unfortunately for our purposes, this characterization does not seem to lead to an efficient algorithm for computing $\Pi(G)$
of planar graphs. 

The simplicity of our characterization comes from the use of what we call a \emph{latching graph}
of a triconnected plane graph instead of 
the traditional tools for reasoning about separators in plane graphs: nooses or radial graphs, which are bipartite plane graphs 
representing the incidences of the vertices and faces (for a precise definition, see \cite{inkmann2008tree} for example).
The latching graph $L_G$ of a biconnected plane graph $G$ is a multigraph
obtained from $G$ by adding, in each face of $G$, every chord of the bounding cycle of the face. 
We observe that $L_G$ is simple if $G$ is triconnected (Proposition~\ref{prop:latching_simple}).
Separators of $G$ correspond to cycles of $L_G$ in a similar way as they
correspond to cycles of the radial graph of $G$. 
Unlike a radial graph, however, that is a bipartite graph on the vertices and the faces of the embedded graph $G$, a latching graph is a graph on $V(G)$ and simpler to reason with for some purposes.
More specifically, we define a class of biconnected plane graphs
we call \emph{steerings} (Definition~\ref{def:steering} in Section~\ref{sec:charact}) 
and show that a vertex set $X$ of a triconnected plane graph $G$ is a PMC if and only if the subgraph of $L_G$ induced by $X$ is a steering. 
This characterization would be less simple if it was formulated in terms of
radial graphs or nooses.
It is not clear if the latching multigraph is useful for reasoning with
plane graphs that are not triconnected.
Adding this notion of latching graphs to the toolbox for triconnected planar graphs is one of our contributions.

\section{Preliminaries} \label{sec:prelim}
\subparagraph{Graphs.}
In this paper, all graphs of main interest are simple, that is, without self loops or parallel edges. Let $G$ be a graph.
We denote by $V(G)$ the vertex set of $G$ and by $E(G)$ the edge set of $G$.
When the vertex set of $G$ is $V$, we say that the graph $G$ is \emph{on} $V$.
The \emph{complete graph} on $V$, denoted by $K(V)$, is a graph on $V$ where every vertex is adjacent to all other vertices.
A complete graph $K(V)$ is a $K_n$ if $|V| = n$.

A graph $H$ is a \emph{subgraph} of a graph $G$ if $V(H) \subseteq V(G)$ and $E(H) \subseteq E(G)$.
The subgraph of $G$ induced by $U \subseteq V(G)$ is denoted by $G[U]$: its  vertex set is $U$ and its edge set is $\{\{u, v\} \in E(G) \mid u, v \in U\}$. 
A vertex set $C \subseteq V(G)$ is a \emph{clique} of $G$ if $G[C]$ is a complete graph.
For each $v \in V(G)$, $N_G(v)$ denotes the set of neighbors of $v$ in $G$: $N_G(v) = \{u \in V(G) \mid \{u, v\} \in E(G)\}$.
The \emph{degree} of $v$ in $G$ is $|N_G(v)|$.
For $U \subseteq V(G)$, the {\em open neighborhood of $U$ in $G$}, denoted by $N_G(U)$, is the set of vertices adjacent to some vertex in $U$ but not belonging to $U$ itself: $N_G(U) = (\bigcup_{v \in U} N_G(v)) \setminus U$. 
For an edge set $E \subseteq E(G)$ we denote by $V(E)$ the set of vertices of $G$ incident to some edge in $E$.

We say that a vertex set $C \subseteq V(G)$ is {\em connected in} $G$ if, for every $u, v \in C$, there is a walk of $G[C]$ starting with $u$ and ending with $v$, where a \emph{walk} of $G$ is a sequence of vertices in which every vertex except for the last is adjacent to the next vertex in the sequence.
It is a {\em connected component} or simply a {\em component} of $G$ if it is connected and is inclusion-wise maximal subject to connectivity.
For each vertex set $S$ of $G$, we denote by $\calC_G(S)$ the set of components of $G[V(G) \setminus S]$.
When $G$ is clear from the context we may drop the subscript and write $\calC(S)$ for $\calC_G(S)$.
A vertex set $S \subseteq V(G)$ is a {\em separator} of $G$ if $|\calC_G(S)| \geq 2$.
We usually assume that $G$ is connected and therefore $\emptyset$ is not a separator of $G$.
A component $C \in \calC_G(S)$ is a \emph{full component} associated with $S$ if $N_G(C) = S$.
A graph $G$ is \emph{biconnected} (\emph{triconnected}) if every separator of $G$ is of cardinality two (three) or greater. 
A graph is a \emph{cycle} if it is connected and every vertex is adjacent to exactly two vertices. 
A graph is a \emph{tree} if it is connected and does not contain a cycle as a subgraph.
A tree is a \emph{path} if it has exactly two vertices of degree one or its vertex set is a singleton.
Those two vertices of degree one are called the \emph{ends} of the path.
If $p$ is a path with ends $a$ and $b$, we say that $p$ is \emph{between} $a$ and $b$.
When we speak of cycles or paths of $G$, we mean subgraphs of $G$ that are cycles or paths. 
Let $p$ be a cycle or path of $G$.
A \emph{chord} of $p$ in $G$ is an edge $\{u, v \}$ of $G$ with $u, v \in V(p)$ that is not an edge of $p$. 
A cycle or path of $G$ is \emph{chordless} if it does not have any chord in $G$.

\subparagraph{Tree decompositions.}
A {\em tree-decomposition} of $G$ is a pair $(T, \calX)$ where $T$ is a tree and $\calX$ is a family $\{X_i\}_{i \in V(T)}$ of vertex sets of $G$, indexed by the nodes of $T$, such that the following three conditions are satisfied. 
We call each $X_i$ the {\em bag} at node $i$.   
\begin{enumerate}
  \item $\bigcup_{i \in V(T)} X_i = V(G)$.
  \item For each edge $\{u, v\} \in E(G)$, there is some $i \in V(T)$ such that $u, v \in X_i$.
  \item For each $v \in V(G)$, the set of nodes $I_v = \{i \in V(T) \mid v \in
  X_i\} \subseteq V(T)$ is connected in $T$.
\end{enumerate}
The {\em width} of this tree-decomposition is $\max_{i \in V(T)} |X_i| - 1$.
The {\em treewidth} of $G$, denoted by $\tw(G)$ is the smallest $k$ such that there is a tree-decomposition of $G$ of width $k$.

\subparagraph{Minimal separators and potential maximal cliques.}
Let $G$ be a graph and $S$ a separator of $G$. 
For distinct vertices $a, b \in V(G)$, $S$ is an {\em $a$-$b$ separator} if there is no path between $a$ and $b$ in $G[V(G) \setminus S]$; it is a {\em minimal $a$-$b$ separator} if it is an $a$-$b$ separator and no proper subset of $S$ is an $a$-$b$ separator.
A separator is a {\em minimal separator} if it is a minimal $a$-$b$ separator for some $a, b \in V(G)$. 
A necessary and sufficient condition for a separator $S$ of $G$ to be minimal is that $\calC_G(S)$ has at least two members that are full components associated with $S$.
We say that a connected vertex set $C$ of $G$ is \emph{minimally separated} if $N_G(C)$ is a minimal separator.

Graph $H$ is {\em chordal} if every cycle of $H$ with four or more vertices has a chord in $H$.
$H$ is a {\em triangulation of graph $G$} if it is chordal, $V(G) = V(H)$, and $E(G) \subseteq E(H)$. 
A triangulation $H$ of $G$ is {\em minimal} if there is no triangulation $H'$ of $G$ such that $E(H')$ is a proper subset of $E(H)$.

A vertex set $X$ of $G$ is a \emph{potential maximal clique} (PMC for short) of $G$ if there is some minimal triangulation $H$ of $G$ such that $X$ is a maximal clique of $H$.
We denote by $\Pi(G)$ the set of all PMCs of $G$.
The following lemmas due to Bouchitt\'{e} and Todinca~\cite{bouchitte2002listing} are
essential in our reasoning about PMCs.
\begin{lemma}[Bouchitt\'{e} and Todinca~\cite{bouchitte2001treewidth}]\label{lem:BT-minsep-PMC}
  Let $X$ be a PMC of a graph $G$. Then, for every component $C$ of $G[V(G) \setminus X]$,
  $S = N_G[C]$ has a full component containing $X \setminus S$ and hence is a minimal separator.
  Moreover, for every minimal separator $S$ of $G$ such that 
  $S \subseteq X$, we have $S \neq X$ and there is some component $C$ of $G[V(G) \setminus X]$ such that
  $S = N_G(C)$.
\end{lemma}

\begin{lemma}[Bouchitt\'{e} and Todinca~\cite{bouchitte2001treewidth}]\label{lem:BT-lemma}
  A vertex set $X$ of graph $G$ is a PMC of $G$ if and only if all of the following conditions hold.
  \begin{enumerate}
  \item There is no full component associated with $X$.
  \item For every $u, v \in X$, either $\{u, v\} \in E(G)$ or there is some $C \in \calC_G(X)$ such that $u, v \in N_G(C)$.
  \end{enumerate}
\end{lemma}

They also show that PMCs are extremely useful for computing the treewidth.
\begin{theorem}
[Bouchitt\'{e} and Todinca \cite{bouchitte2001treewidth}]
\label{thm:bt-dp}
    Given a graph $G$ of $n$ vertices and $\Pi(G)$, the treewidth of $G$ can be computed in time $|\Pi(G)|n^{O(1)}$.
\end{theorem}

We need the following property of PMCs.
\begin{lemma}
[Bouchitt\'{e} and Todinca~\cite{bouchitte2001treewidth}]
\label{lem:pmc_subset}
Let $X$ be a PMC of a graph $G$.
Then, every proper subset of $X$ has a full component associated with it. As a consequence, no proper subset of a PMC is a PMC.
\end{lemma} 

\subparagraph{Plane graphs.}
A \emph{sphere-embedded graph} is a graph where a vertex is a point on the sphere $\sphere$ and each edge is a simple curve on $\sphere$ between two distinct vertices that does not contain any vertex except for its ends. 
In the combinatorial view of a sphere-embedded graph $G$, an edge represents the adjacency between two vertices at its ends. 
A sphere-embedded graph $G$ is a \emph{plane graph} if no two edges of $G$ intersect each other as curves except at their ends.
A graph $G$ is \emph{planar} if 
there is a plane graph that is isomorphic to $G$ in the combinatorial view.
We call this plane graph a \emph{planar embedding} of $G$.

Let $G$ be a biconnected plane graph. 
Then the removal of all the edges of $G$ from the sphere $\Sigma$ results in a collection of regions homeomorphic to open discs.
We call these regions the \emph{faces} of $G$.
Each face of $G$ is bounded by a simple closed curve that consists of edges of $G$.
A vertex or an edge of $G$ is \emph{incident} to a face $f$ of $G$, if it is contained in the closed curve
that bounds $f$. 
We denote by $F(G)$ the set of faces of $G$. 
For each $f \in F(G)$, we denote by $V(f)$ the set of vertices of $G$ incident to $f$.

A \emph{combinatorial} representation of a plane graph $G$ is an indexed set $\{\pi_v\}_{v \in V(G)}$, where $\pi_v$ is a permutation on $N_G(v)$ such that $\pi_v(u)$ for $u \in N_G(v)$ is the neighbor of $v$ that comes immediately after $u$ in the clockwise order around $v$. 
Suppose we have two planar embeddings $G_1$ and $G_2$ of a graph $G$ and assume that $V(G_1) = V(G_2) = V(G)$ and moreover the isomorphisms between $G_1$ and $G$ as well as between $G_2$ and $G$ are identities. 
Let, for $i \in \{1,2\}$, the combinatorial representation of $G_i$ be $\{\pi_{i, v}\}_{v \in V(G)}$. 
We say that these two planar embeddings are \emph{combinatorially equivalent} if either $\pi_{1, v} = \pi_{2, v}$ for every $v \in V(G)$ or $\pi_{1, v} = (\pi_{2, v})^{-1}$ for every $v \in V(G)$.
In this paper, we mainly deal with triconnected planar graphs.
It is known that, for a triconnected planar graph $G$, the planar embedding of $G$ is unique up to combinatorial equivalence 
\cite{Whitney1933}. 
For this reason, we will be dealing with triconnected plane graphs rather
than triconnected planar graphs.

In conventional approaches for branch-decompositions and tree-decompositions of
plane graphs \cite{seymour1994call,bouchitte2003chordal}, the standard tools are radial graphs and nooses.
We observe that, for triconnected plane graphs, latching graphs as defined below may replace those tools and greatly simplify definitions and reasoning.

\begin{definition}
    \label{def:latching}
    Let $G$ be a biconnected plane graph.
    The \emph{latching graph} of $G$, denoted by $L_G$, is a multigraph obtained from $G$ by, for each face $f$ bounded by a cycle of four or more vertices, drawing every chord of the cycle within $f$. 
\end{definition}
Figure~\ref{fig:latching} shows a triconnected plane graph with
its latching graph and a plane graph 
for which the latching graph is not simple.

\begin{proposition}
    \label{prop:latching_simple}
The latching graph of a plane graph $G$ is simple if $G$ is triconnected.
\end{proposition}
\begin{proof}
    Let $G$ be a triconnected plane graph and let $u$ and $v$ be two vertices of $G$.
    It is straightforward to see, and is observed in \cite{bouchitte2003chordal},
that there are exactly two faces of $G$ to which both $u$ and
$v$ are incident if $\{u, v\}$ is an edge of $G$ and
there is at most one such face otherwise. In the first case,
there is no face in which $\{u, v\}$ is a chord and, in the second case,
there is at most one face in which $\{u, v\}$ is a chord.
Thus, there is at most one edge between $u$ and $v$ in $L_G$ 
and therefore $L_G$ is simple.
\end{proof}

\begin{figure}[hbtp]
\begin{center}
\includegraphics[width=12cm]{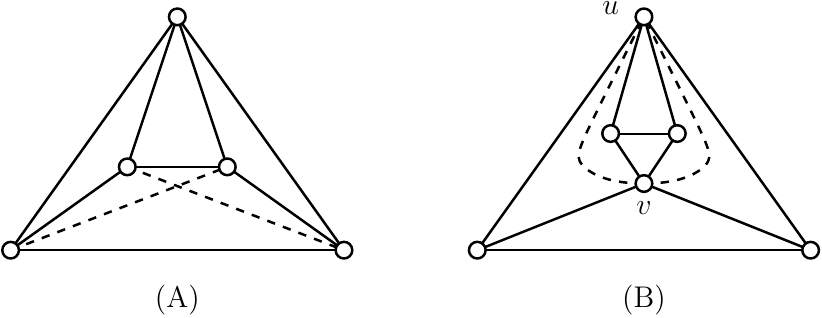}
\caption{(A) a triconnected plane graph (solid edges) and its latching graph (solid and broken edges).
(B) edge $\{u, v\}$ is drawn in more than one face if $u$ and $v$ separate the graph.}
\label{fig:latching}
\end{center}
\end{figure}

Although the latching graph $L_G$ of a triconnected plane graph $G$ is not a plane graph in general, we are interested in its subgraphs that are plane graphs. For $X \subseteq V(G)$, the \emph{subgraph of $L_G$ induced by $X$}, denoted by $L_G[X]$, is a sphere-embedded graph with vertex set $X$ that inherits all edges of $L_G$, as curves, with two ends in $X$.

\begin{proposition}
    \label{prop:plane-latch-subgraph}
    Let $G$ be a triconnected plane graph and let $X$ be a vertex set of $G$.
    Then, $L_G[X]$ is a plane graph if and only if there is no face $f$ of $G$ such that $|V(f) \cap X| \geq 4$.
\end{proposition}
\begin{proof}
Since $V(f) \cap X$ forms a clique of $L_G[X]$ for each face $f$ of $G$, $L_G[X]$ is not a plane graph if 
there is some face $f$ such that $|V(f) \cap X| \geq 4$. Conversely, if $L_G[X]$ has an edge-crossing,
it must be in some face $f$ of $G$ and we have $|V(f) \cap X| \geq 4$.
\end{proof}

Let $X \subseteq V(G)$ such that $L_G[X]$ is a biconnected plane graph.
We call each face of $L_G[X]$ a \emph{region} of $L_G[X]$, in order to avoid confusions with the faces of $G$.
We say that a region $r$ of $L_G[X]$ is \emph{empty} if $r$ does not contain any vertex of $G$. 

\begin{proposition}
  \label{prop:compo_in_region}  
  Let $G$ be a triconnected plane graph and let $X \subseteq V(G)$ be such that $L_G[X]$ is a biconnected plane graph.  Then, each $C \in \calC_G(X)$ is contained in some region of $L_G[X]$.
  Moreover, each region of $L_G[X]$ contains at most one component in $\calC_G(X)$ and, if
  the region is incident to four or more vertices, exactly one component.
\end{proposition} 
\begin{proof}
Since no edge of $G$ crosses any edge of $L_G$, each edge of $G[V(G) \setminus X]$ has
its ends in the same region of $L_G[X]$. Therefore, for each $C \in \calC_G(X)$,
all the vertices of $C$ lie in the same region of $L_G[X]$. 
For the second statement, suppose a region $r$ of $L_G[X]$ contains more than
one component in $\calC_G(X)$. 
Then, there is some $x \in X$ on the boundary of
$r$ that is adjacent to vertices from more than one component in $\calC_G(X)$ contained in $r$.
Let $e_1$, $e_2$, \ldots, $e_m$ be the edges incident to $x$ lying in $r$, arranged in the
clockwise ordering, and let $v_i$, $i \in \{1, \ldots, m\}$,
be the other end of edge $e_i$. Since $r$ is a region of $L_G[X]$, none of $v_i$, 
$i \in \{1, \ldots, m\}$, belongs to $X$. Therefore, for some $i$, $1 \leq i < m$,
$v_i$ and $v_{i + 1}$ belong to distinct components in $\calC_G(X)$. 
But the face of $G$ to which $x$, $v_i$, and $v_{i + 1}$ are incident is not incident
to a vertex in $X$ other than $x$, since, if it did, the cycle bounding $r$ would have a chord
in $L_G[X]$ that separates $r$, a contradiction. Therefore, there is a path of $G[V(G) \setminus X]$
between $v_i$ and $v_{i+1}$ around this face, also a contradiction.
We conclude that $r$ contains at most one component in $\calC_G[X]$.

Finally, suppose a region $r$ of $L_G[X]$ is empty. Then $V(r) \subseteq V(f)$ for some face $f$ of $G$.
We have $|V(r)| = 3$, since otherwise $L_G[X]$ would not be a plane graph due to Proposition~\ref{prop:plane-latch-subgraph}.
Therefore, every region $r$ with $|V(r)| \geq 4$ is non-empty and hence contains exactly one component.
\end{proof}

We use this framework to formulate the fundamental observation on minimal separators of triconnected plane graphs, which is conventionally stated in terms of nooses~\cite{bouchitte2003chordal}.
\begin{proposition}
\label{prop:min-sep-of-triconn}
Let $G$ be a triconnected plane graph.
A vertex set $S$ of $G$ is a minimal separator of $G$ if and only if $L_G[S]$ is a cycle and 
neither of the two regions of $L_G[S]$ is empty.
\end{proposition}
\begin{proof}
Let $S$ be a minimal separator of $G$ and let $C_1, C_2 \in \calC_G(S)$ be distinct such that
$N_G(C_i) = S$ for $i \in \{1, 2\}$. To separate $C_1$ from $C_2$, $L_G[S]$ must contain a cycle
$\gamma$ that separates this components. Since $S$ is a minimal separator, we indeed have
$V(\gamma) = S$ and $L_G[S]$ is a biconnected plane graph.
Due to Proposition~\ref{prop:compo_in_region},
there are regions $r_1$ and $r_2$ of $L_G[S]$ such that $C_i$ is the unique component in 
$\calC_G(X)$ that is contained in $r_i$, for $i \in \{1, 2\}$.
Since $S$ is a minimal separator separating $C_1$ from $C_2$, every vertex in $S$ is incident
to both regions $r_1$ and $r_2$ and therefore the cycles of $L_G[S]$ bounding these two regions
are identical to each other and to $L_G[S]$. Therefore, $L_G[S]$ is a cycle.

For the converse, suppose $L_G[S]$ is a cycle and neither of the two regions $r_1$ and $r_2$ 
of $L_G[S]$ is empty. Let $C_i$ be the set of vertices of $G$ contained in $r_i$ for $i \in \{1, 2\}$.
Due to Proposition~\ref{prop:compo_in_region},
both $C_1$ and $C_2$ belongs to $\calC_G(S)$.
We claim that $N_G(C_1) = S$. This is trivial if $|S| = 3$ 
since $G$ is triconnected. So suppose $|S| \geq 4$. Let $s \in S$ be arbitrary and
let $s_1$ and $s_2$ be the vertices on the cycle $L_G[S]$ adjacent to $s$.
If there is no edge of $G$ between $s$ and vertices in $C_1$ then $s_1$ and
$s_2$ belong to the same face of $G$ and therefore there is an edge of $L_G$ between
$s_1$ and $s_2$. Since $s_1, s_2 \in S$ this edge belongs to $L_G[S]$ contradicting
the assumption that $L_G[S]$ is a cycle. This shows that $N_G(C_1) = S$.
As we similarly have $N_G(C_2) = S$, $S$ is a minimal separator of $G$.
\end{proof}
\begin{remark}
It follows from this proposition that every minimal separator of a triconnected planar graph has exactly two full components associated with it.
\end{remark}

\begin{corollary}
\label{cor:minsep4}
   Let $G$ be a triconnected plane graph.
A vertex set $S$ of $G$ with $|S| \geq 4$ is a minimal separator of $G$ if and only if
$L_G[S]$ is a cycle. 
\end{corollary}
\begin{proof}
    Suppose $|S| \geq 4$ and $L_G[S]$ is a cycle. 
    Then, neither of the two regions of $L_G[S]$ is empty, since if one of the two regions $r$ is empty, $r$ is contained in a face $f$ of $G$ and the definition of $L_G[S]$ implies that $L_G[S]$ is a complete graph, not a cycle.
    Therefore, the ``if'' part of Proposition~\ref{prop:min-sep-of-triconn} applies and 
    $S$ is a minimal separator of $G$.
    The other direction is an immediate consequence of the ``only if'' part of Proposition~\ref{prop:min-sep-of-triconn}.
\end{proof}

The following fact is essential in our treatment of PMCs of triconnected plane graphs.
\begin{proposition}
\label{prop:pmc-planar}
    For every PMC $X$ of a triconnected plane graph $G$, $L_G[X]$ is a biconnected plane graph.
    Moreover, every region of $L_G[X]$ is bounded by a chordless cycle of $L_G$.
\end{proposition}
\begin{proof}
Let $X$ be a PMC of $G$.
We first show that $L_G[X]$ is biconnected. 
Consider an arbitrary vertex $x \in X$ in the PMC.
We show that $L_G[X \setminus \{x\}]$ is connected.
Let $y$ and $z$ be any two members of $X \setminus \{x\}$.
Due to the second condition for PMCs in Lemma~\ref{lem:BT-lemma}, either $y$ is adjacent to $z$ in $G$ or there is some $C \in \calC_G(X)$ such that $y, z \in N_G(C)$.
In the former case, $y$ is adjacent to $z$ in $L_G[X]$. 
In the latter case, $S = N_G(C)$ is a minimal separator of $G$ and, therefore, $L_G[S]$ is a cycle due to Proposition~\ref{prop:min-sep-of-triconn}.
It follows that there is a subpath of the cycle $L_G[S]$ between $y$ and $z$ that does not contain $x$.
Therefore, $L_G[X \setminus \{x\}]$ is connected for every $x \in X$ and hence $L_G[X]$ is biconnected.

To show that $L_G[X]$ is a plane graph, suppose that it is not the case. 
Due to Proposition~\ref{prop:plane-latch-subgraph},  there is a face $f$ of $G$ such that $|V(f) \cap X| \geq 4$.
    Let $x_1$, $x_2$, $x_3$, and $x_4$ be arbitrary four vertices in $V(f)$, listed in the order around $f$. 
    If $x_1$ is not adjacent to $x_3$ in $G$, there is some $C \in \calC_G(X)$ such that $x_1, x_3 \in N_G(C)$. 
    In either case, there is a path $p_1$ of $G$ between $x_1$ and $x_3$ that intersects $X$ only at $x_1$ and $x_3$.
    Similarly, there is a path $p_2$ of $G$ between $x_2$ and $x_4$ that intersects $X$ only at $x_2$ and $x_4$.
    Since the Jordan curve obtained from $p_1$ by adding a curve in $f$ separates $x_2$ and $x_4$, $p_1$ and $p_2$ must intersect at a vertex of $G$.
    Therefore, there is some component $C \in \calC_G(X)$ such that $x_i \in N_G(C)$ for every $i \in \{1, 2, 3, 4\}$.
    But $S = N_G(C)$ cannot be a minimal separator, since $|S \cap V(f)| \geq 4$ and therefore $L_G[S]$ is not a cycle. 
    This contradicts the assumption that $X$ is a PMC, because of Lemma~\ref{lem:BT-minsep-PMC}. 
    Therefore, $L_G[X]$ is a plane graph.

    Let $r$ be an arbitrary region of $L_G[X]$. 
    If $r$ is non-empty and $C$ is the set of vertices lying in $r$, then $C \in \calC_G(X)$ 
    due to Proposition~\ref{prop:compo_in_region} and hence $S = N_G(C)$ is a minimal separator of $G$ 
    due to Lemma~\ref{lem:BT-minsep-PMC} since $X$ is a PMC of $G$ .
    Due to Proposition~\ref{prop:min-sep-of-triconn}, $L_G[S]$ is a cycle and bounds $r$.
    Suppose, on the other hand, that $r$ is empty. Then the cycle bounding $r$ has exactly three vertices 
    because otherwise $G$ has a face $f$ such that $|V(f) \cap X| \geq 4$ and $L_G[X]$ would not be a plane graph
    due to Proposition~\ref{prop:plane-latch-subgraph}. This cycle, having exactly three vertices, is trivially chordless.
\end{proof}

\section{Characterizing PMCs in triconnected planar graphs}
\label{sec:charact}

In this section, we characterize PMCs of a triconnected plane graph
in terms of graphs we call steerings.

\begin{definition}
\label{def:steering}
Let $\gamma$ be a cycle. A subset $R$ of $V(\gamma)$ is a \emph{slot} of
$\gamma$ if $R$ is a singleton or an edge of $\gamma$.
A graph $H$ is a \emph{steering}, if there is a bipartition $(S, P)$ of
$V(H)$ such that $H[S]$ is a cycle, 
$N_H(P)$ is neither empty nor a slot of $H[S]$, 
and if $|P| \geq 2$ then the following conditions hold.
\begin{enumerate}
    \item $H[P]$ is a path. 
    \item No internal vertex of the path $H[P]$ is adjacent to any vertex in $S$.
    \item For each end $t$ of the path, $N_H(t) \cap S$ is a slot of $H[S]$.
\end{enumerate}
We call $H$ an \emph{$(S, P)$-steering} in this situation.
We call a steering a \emph{wheel} if it is an $(S, P)$-steering for some
bipartition $(S, P)$ of $V(H)$ such that $|P| = 1$; otherwise it is a \emph{non-wheel steering}.
\end{definition}

Figure~\ref{fig:steerings} shows some steerings. 
Figure~\ref{fig:steerings_bipartitions} shows two steerings with different $(S, P)$-bipartitions. The example (B)
is an $(S, P)$-steering for the second bipartition but not for the first bipartition.
Figure~\ref{fig:non-steering} shows some graphs that are not steerings. 
\begin{figure}[hbtp]
\begin{center}
\includegraphics[width=14cm]{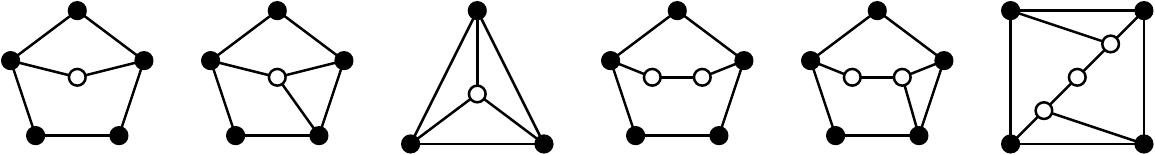}
\caption{Some steerings. White vertices belong to $P$ and black vertices belong to $S$, in a choice of
the bipartition $(S, P)$ that works. The first three are wheels while the remaining three are non-wheels.}
\label{fig:steerings}
\end{center}
\end{figure}

\begin{figure}[hbtp]
\begin{center}
\includegraphics[width=10cm]{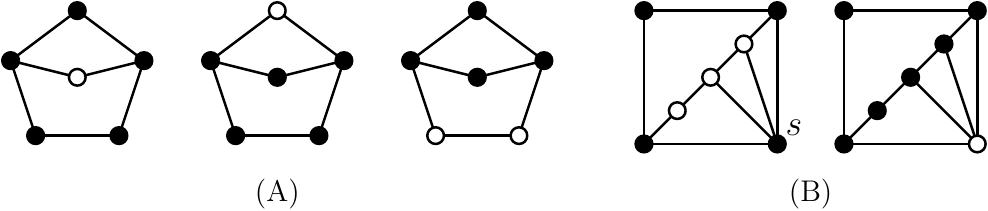}
\caption{(A) A steering with three possible $(S, P)$ bipartitions. Because of the first or the second bipartition,
it is a wheel. (B) Also a wheel. For the first $(S, P)$ bipartition, it is not an $(S, P)$-steering since 
an internal vertex of the path on $P$ are adjacent to $s \in S$. The second $(S, P)$ bipartition shows that
it is a wheel.}
\label{fig:steerings_bipartitions}
\end{center}
\end{figure}

\begin{figure}[hbtp]
\begin{center}
\includegraphics[width=8cm]{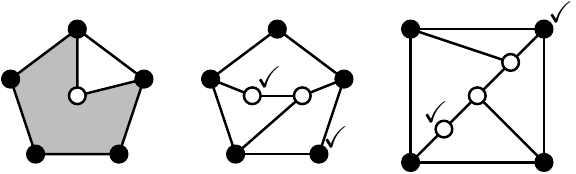}
\caption{Some graphs that are not steerings. In the shown $(S, P)$ bipartition,
the first example $H$ is not a steering since $N_H(P) \cap S$ is a slot.
The second example is not a steering since $N_H(t) \cap S$ is not a slot for the right
end $t$ of the path $H[P]$. The third example is not a steering since an internal vertex of
the path $H[P]$ is adjacent to a vertex in $S$. It can be verified that these graphs
are not $(S, P)$-steerings for any other $(S, P)$ bipartition.
It is explained in the main text why any of these plane graphs cannot be
$L_G[X]$ for a PMC $X$ of a triconnected plane graph $G$.}
\label{fig:non-steering}
\end{center}
\end{figure}

\begin{proposition}
    \label{prop:steering-planar}
    Let $H$ be a steering. Then $H$ is biconnected and planar.
    Moreover, the planar embedding of $H$ is unique up to combinatorial equivalence.
\end{proposition}
\begin{proof}
    Let $H$ be a $(S, P)$-steering. We first confirm that $H$ is biconnected. 
    Let $v$ be an arbitrary vertex of $H$. If $v \in S$, then $H[S \setminus \{v\}]$ is a path and
    $N_H(P) \cap (S \setminus \{v\})$ is non-empty as $N_H(P) \cap S$ is neither empty nor a singleton, so $H[V(H) \setminus \{v\}]$ is connected.
    If $v \in P$, then each component of $H[P \setminus \{v\}]$ has a vertex adjacent to a vertex in $S$ and
    hence $H[V(H) \setminus \{v\}]$ is connected.

    We next show that $H$ is planar and has a unique planar embedding up to combinatorial equivalence.
 Suppose first that $|N_H(P)| = 2$. Then the two members $s_1$ and $s_2$ of $N_H(P)$ are of degree three and
 all other vertices of $H$ are of degree two. We draw three paths between $s_1$ and $s_2$ without crossings
 and hence $H$ is planar. The combinatorial embedding of $H$ in this case essentially consists of the clockwise ordering
 $\pi_1$ of $N_H(s_1)$ around $s_1$ and the clockwise ordering $\pi_2$ of $N_H(s_2)$ around $s_2$.
 For each $i \in \{1, 2\}$, we have two choices of $\pi_1$, one being the inverse of the other, 
 and if we flip $\pi_1$ then $\pi_2$ is flipped.
 Therefore, there are only two possible combinatorial embeddings of $H$, which are equivalent to each other.
 
 Suppose next that $|N_H(P)| \geq 3$. The planar drawing of $H$ can be obtained by first drawing the
 cycle $H[S]$, draw the path $H[P]$ in one of the faces of $H[S]$, and then
 drawing edges between vertices in $P$ and the vertices in $S$. If $|P| = 1$ then 
 these edges can certainly be drawn without crossings. Suppose $|P| \geq 2$ and
 let $t_1$ and $t_2$ be the two ends of the path $H[P]$. By the definition
 of an $(S, P)$-steering, $N_H[t_i] \cap S$ is a slot of the cycle $H[S]$ for $i \in \{1, 2\}$.
 A crossing of edges would be unavoidable only if $N_H[t_1] \cap S$ and $N_H[t_1] \cap S$
 are an identical edge of $H[S]$. But this is impossible since $N_H(P)$ would then be equal to this edge
 and hence would be a slot, violating the condition for $H$ to be an $(S, P)$-steering.
 In this case, where we are assuming that $|N_H(P)| \geq 3$, there is a vertex $v \in P$
 of degree $\geq 3$ in $H$. If $|P| = 1$ then $v$ is the unique member of $P$ and
 if $|P| \geq 2$ then $v$ is one of the ends $t_1$ and $t_2$ of the path $H[P]$.
 In either case $v$ is adjacent to at least two vertices in $S$. Since the ordering of
 $N_H(v)$ is constrained by the ordering of $S$ along the cycle $H[S]$, we have
 only two choices for the clockwise ordering $\pi_v$ of $N_H(v)$ around $v$, one being the inverse of the other.
 When we choose one of them, then the clockwise order of $S$ around the face containing $P$
 is determined and this in turn determines the clockwise ordering $\pi_u$ of
 $N_H(u)$ around $u$ for every vertex $u$ of degree $\geq 3$ in $H$.
 When we flip the choice of $\pi_v$ then $\pi_u$ is flipped for every vertex $u$ of degree $\geq 3$ in $H$.
 Therefore, there are only two combinatorial embeddings of $H$, which are equivalent to each other.
\end{proof}

From now on, when we refer to a steering $H$, we view $H$ as a plane graph rather than a combinatorial graph as originally defined.

The following lemma shows that $X$ is a PMC of a triconnected plane graph $G$ if $L_G[X]$ is a steering. Before going into technical details, it may be helpful to study the examples in Figures~\ref{fig:steerings} and \ref{fig:steerings_bipartitions}
and confirm the following.
\begin{enumerate}
\item The cycle bounding each face is chordless.
\item Every pair of vertices is either adjacent to each other or incident to a common face.
\end{enumerate}
Also observe that if $L_G[X]$ is one of the non-steering graphs 
in Figure~\ref{fig:non-steering}, then $X$ cannot be a PMC of $G$:
in the first example, the component of $G[V(G) \setminus X]$ lying in the gray face would be a full component of $X$, violating the first condition of Lemma~\ref{lem:BT-lemma} for $X$ to be a PMC of $G$;
in the second and third examples, the pair of checked vertices do not share a common face and therefore cannot 
belong to the neighborhood of a common component of $G[V(G) \setminus X]$, violating the second condition of
Lemma~\ref{lem:BT-lemma} for $X$ to be a PMC of $G$.
\begin{lemma}
\label{lem:steering-is-PMC}
    Let $X$ be a vertex set of a triconnected plane graph $G$ and
    suppose $L_G[X]$ is a steering. Then, $X$ is a PMC of $G$.
\end{lemma}

\begin{proof} 
Suppose $L_G[X] = H$ where $H$ is an $(S, P)$-steering.
This equality means an equality as plane graphs, since we are viewing a steering as a biconnected plane graph, 
as we mentioned in the remark following Proposition~\ref{prop:steering-planar}.

Due to Proposition~\ref{prop:compo_in_region}, each component in $\calC_G(X)$ is contained in a region of
$H$ and each region of $H$ contains at most one component in $\calC_G(X)$.

    We show that $X$ satisfies the two conditions for PMCs in Lemma~\ref{lem:BT-lemma}:
    (1) There is no full component associated with $X$;
    (2) For every pair of vertices $x$ and $y$ in $X$,
    either $x$ is adjacent to $y$ in $G$ or there is some $C \in \calC_G(X)$
    such that $x, y \in N_G(C)$.
    
    To show that Condition (1) holds, let $C$ be an arbitrary component in $\calC_G(X)$ and
    $\gamma_C$ be the cycle bounding the region of $H$ containing $C$.
    If $\gamma_C$ is $H[S]$, then $C$ is not a full component associated with $X$
    as $S$ is a proper subset of $X$. Otherwise, there is a vertex in $S$ that does not belong to
    $\gamma_C$, since $N_H(P)$ is not a slot of the cycle $H[S]$, and therefore $C$ is not a full component 
    associated with $X$.

    To show that Condition (2) holds, 
    we first show that, for arbitrary two members $x$ and $y$ of $X$,
    there is some region of $H$ whose boundary contains both $x$ and $y$.
    If $x, y \in S$, then this certainly holds since $H$ has a region bounded by the cycle $H[S]$.
    So suppose that at least one of $x$ and $y$, say $x$, belongs to $P$.
    We argue in a few cases. 
    
    First suppose that $|P| = 1$ and hence $P = \{x\}$ and $y \in S$.
    Then, there is certainly a region of $H$ bounded by a cycle consisting of $x$ and a subpath
    of the cycle $H[S]$ that contains $y$. See the first three examples in Figure~\ref{fig:steerings}.

    Next suppose that $|P| \geq 2$. Since there is no internal vertex of the path $H[P]$ adjacent
    in $H$ to any vertex in $S$, there are two regions of $H$ bounded by cycles
    containing the path $H[P]$. Let $\gamma_1$ and $\gamma_2$ be those cycles. 
    Then, $\gamma_i$, for each $i \in \{1, 2\}$, consists of $H[P]$ and a subpath $p_i$
    of $H[S]$. Since $N_H(t) \cap S$ for each end $t$ of the path $H[P]$ is a slot,
    we have $V(p_1) \cap V(p_2) = S$. Therefore, either $\gamma_1$ or $\gamma_2$ contains
    both $x$ and $y$. See the last three examples in Figure~\ref{fig:steerings}.

    We are ready to show that Condition (2) holds. 
    Let $x$ and $y$ be arbitrary two vertices 
    in $X$. As we have shown above, there is a cycle $\gamma$ bounding a region
    of $H$ such that $x, y \in V(\gamma)$. If this region contains some
    $C \in \calC_G(X)$, then we are done since $N_G(C) = V(\gamma)$.
    Otherwise, $\gamma$ bounds an empty region and must be a triangle 
    since otherwise $L_G[X]$ would not be a plane graph. 
    Therefore, $\{x, y\}$ is an edge of $L_G[X]$. If this is an edge of $G$ then
    we are done, so suppose not. Let $\gamma'$ be another
    cycle bounding a region of $L_G[X]$ that contains the edge $\{x, y\}$.
    This region bounded by $\gamma'$ cannot be empty since, if it were,
    then we would have a face of $G$ incident
    to four or more vertices of $X$, and $L_G[X]$ would not be a plane graph.
    We are done, since $x, y \in N_G(C')$ where $C'$ is the member of $\calC_G(X)$ contained
    in the region bounded by $\gamma'$.
\end{proof}

The converse of this lemma is shown in several steps. 
We start with a technical proposition.

\begin{proposition}
\label{prop:subset-enough}
Let $X$ be a PMC of a triconnected plane graph $G$.
If $L_G[X']$ is a steering for some $X' \subseteq X$, then
$X = X'$ and hence $L_G[X]$ is a steering.
\end{proposition}
\begin{proof}
    If $L_G[X']$ is a steering, $X'$ is a PMC of $G$
    due to Lemma~\ref{lem:steering-is-PMC}.
    Since no proper subset of a PMC is a PMC (Lemma~\ref{lem:pmc_subset}), 
    we have $X = X'$.
\end{proof}

The next lemma deals with a special case where every minimal separator contained in a PMC $X$
consists of three vertices.

\begin{lemma}
    \label{lem:all-triang}
    Let $X$ be a PMC of a triconnected plane graph $G$ and suppose
    that, for every $C \in \calC_G(X)$, $|N_G(C)| = 3$.
    Then, $L_G[X]$ is a $K_4$.
\end{lemma}
\begin{proof}
Recall that, due to Proposition~\ref{prop:pmc-planar}, $L_G[X]$ is a
biconnected plane graph. Recall also that every $C \in \calC_G(X)$
is contained in a region of $L_G[X]$ bounded by a cycle whose vertex set
is $N_G(C)$.
    Let $\gamma$ be a cycle that bounds a region $r$ of $L_G[X]$ and
    let $S = V(\gamma)$. If $r$ is empty then $|S| = 3$ since
    otherwise $L_G[X]$ would not be a plane graph (Proposition~\ref{prop:plane-latch-subgraph}). 
    Otherwise $r$ contains
    some $C \in \calC_G(X)$ and we also have $|S| = 3$ due to the assumption.
    Let $x \in X \setminus S$ and $s \in S$ be arbitrary.
    We show that $x$ is adjacent to $s$ in $L_G[X]$.
    If $x$ is adjacent to $s$ in $G$ then we are done.
    Otherwise, due to the second condition for PMCs in Lemma~\ref{lem:BT-lemma}, 
    there is some $C \in \calC_G(X)$ such that
    $x, s \in N_G(C)$. Since the region of $L_G[X]$ containing $C$
    is bounded by a triangle on $N_G(C)$, $x$ is adjacent to
    $s$ in $L_G[X]$.
    Therefore, $x$ is adjacent in $L_G[X]$ to every $s \in S$
    and hence $L_G[S \cup \{x\}]$ is a $K_4$.
    As $L_G[S \cup \{x\}]$ is an $(S, \{x\})$-steering, we have 
    $X = S \cup \{x\}$ due to Proposition~\ref{prop:subset-enough}
    and hence $L_G[X]$ is a $K_4$.
\end{proof}

To deal with the more general case, we use the following notion of \emph{arches}.
\begin{definition}
Let $G$ be a triconnected plane graph and let
$S$ be a minimal separator of $G$ with $|S| \geq 4$. 
An \emph{arch} of $S$ is a subset $P$ of $V(G) \setminus S$
such that $L_G[P]$ is a path and 
$N_{L_G}(P) \cap S$ is neither empty nor a slot of the cycle $L_G[S]$.
\end{definition}
If $L_G[S \cup P]$ is an $(S, P)$-steering, then $P$ is certainly an arch of $S$.
The converse does not hold: if $P$ is an arch of $S$ then $L_G[S \cup P]$
is not necessarily an $(S, P)$-steering, as the definition of steerings requires 
more conditions to be satisfied by $L_G[S \cup P]$.
We show, however, in the next lemma that if $P$ is an inclusion-wise minimal arch
then $L_G[S \cup P]$ is a steering. A caveat: $L_G[S \cup P]$ is not necessarily an
$(S, P)$-steering in this case; it may be an $(S', \{s\})$-steering where $S' = (S \cup P) 
\setminus \{s\}$ for some $s \in S$. In Lemma~\ref{lem:pmc-has-arch}, we show that
if $X$ is a PMC then every minimal separator $S \subset X$ with $|S| \geq 4$ has 
an arch $P$ that is a subset of $X \setminus S$. It follows from these two lemmas
that, for every minimal separator $S$ with $|S| \geq 4$ that is a subset of a PMC $X$, 
there is some $P \subset X \setminus S$ such that $L_G[S \cup P]$ is a steering.
Due to Lemma~\ref{lem:steering-is-PMC}, $S \cup P$ is a PMC and, since
no proper subset of a PMC is a PMC (Lemma~\ref{lem:pmc_subset}), we conclude
that $X = S \cup P$ and $L_G[X]$ is a PMC, which is an essential part of
Theorem~\ref{thm:PMC-is-steering}.

\begin{lemma}
\label{lem:min-arch2steering}
Let $S$ be a minimal separator of a triconnected plane graph
$G$ with $|S| \geq 4$.
Suppose $S$ has an arch and let $P$ be 
an inclusion-wise minimal arch of $S$. 
Let $H = L_G[S \cup P]$.
Then, either $H$ is an $(S, P)$-steering 
or there is some $s \in N_{L_G}(P) \cap S$ such that
$H$ is an $((S \cup P) \setminus \{s\}, \{s\})$-steering.
\end{lemma}
\begin{proof}
If there is some $v \in P$ such that $N_{L_G}(v) \cap S$ is neither empty nor a slot of
the cycle $L_G[S]$, then $\{v\}$ is an arch of $S$ and, since $P$ is minimal, $P = \{v\}$.
We are done, since $H$ is an $(S, \{v\})$-steering.

Suppose $N_{L_G}(v) \cap S$ is either empty or a slot of $L_G[S]$ for every $v \in P$.
Since $N_{L_G}(P) \cap S$ is not a slot of $L_G[S]$, there are two vertices $v_1, v_2 \in P$
and two vertices $s_1, s_2 \in S$ such that $v_i$ is adjacent to $s_i$ for $i \in \{1, 2\}$
and $\{s_1, s_2\}$ is not an edge of $S$. Let $p$ be a shortest path in $L_G[P]$ between
$v_1$ and $v_2$. Then, $V(p)$ is an arch of $S$ and, since $P$ is minimal, $P$ must be equal to 
$V(p)$. If no internal vertex of $p$ is adjacent in $L_G$ to any vertex in $S$, then 
$L_G[S \cup P]$ is an $(S, P)$-steering and we are done.
So suppose an internal vertex $v$ of $p$ is adjacent to some vertex in $S$.
Let $p_i$ be the subpath of $p$ between $v_i$ and $v$, for $i \in \{1, 2\}$.
Since $P$ is a minimal arch of $S$, neither $V(p_1)$ nor $V(p_2)$ is an arch of
$S$. Therefore, $N_{L_G}(\{v_i, v\}) \cap S$ is a slot for $i \in \{1, 2\}$.
This is possible only if there is a vertex $s \in S$ such that $s$ is adjacent to both $s_1$ and $s_2$
on $L_G[S]$, $N_{L_G}(v_i)$ is either $\{s_i\}$ or $\{s_i, s\}$ for $i \in \{1, 2\}$, and 
$N_{L_G}(v) \cap S = \{s\}$. Moreover, since $P = V(p)$ is a minimal arch of $S$, no vertices in $P$
other than $v$, $v_1$, or $v_2$ are adjacent to any vertex in $S$. Let $p_3$ be the subpath of 
$L_G[S]$ between $s_1$ and $s_2$ avoiding $s$ and let $\gamma$ be the cycle consisting of $p_3$
and $p$. 
Let $S' = V(\gamma) = (S \cup P) \setminus \{s\}$. We claim that $\gamma$ is chordless in 
$L_G$, that is, $L_G[S'] = \gamma$. This is because $p_3$ is chordless since $L_G[S]$ is 
a cycle by assumption, $p$ is chordless since it is the shortest path in $P$ between $v_1$ and $v_2$,
and the only edges between $V(p)$ and $V(p_3)$ are the edges $\{v_1, s_1\}$ and $\{v_2, s_2\}$.
Since $N_{L_G}(s) \cap S' = \{s_1, v, s_2\}$, $L_G[S \cup P]$ is an $(S', \{s\})$-steering as desired.
\end{proof}

\begin{lemma}
    \label{lem:pmc-has-arch}
    Let $X$ be a PMC of a triconnected plane graph $G$ and
    let $S \subseteq X$ be a minimal separator of $G$ such that $|S| \geq 4$. 
    Then, there is an arch $P$ of $S$ such that $P \subseteq X \setminus S$.  
\end{lemma}
\begin{proof}
Due to \cref{lem:BT-minsep-PMC}, $X \setminus S$ is non-empty and
is a subset of one of the two full components associated with $S$.
Let $r_0$ be the region of $L_G[S]$ that contains $X \setminus S$.

Suppose that there is no arch $P$ of $S$ such that $P \subseteq X \setminus S$. 
Then, for every pair of vertices $s$ and $s'$ in $S$ that are not adjacent to each other 
on the cycle $L_G[S]$, we can draw a curve between $s$ and $s'$ in $r_0$ without crossing
the drawing of $L_G[X]$. Therefore, there is a region $r$ of $L_G[X]$ contained in $r_0$ that is incident to all vertices in $S$. 
Let $\gamma$ be the cycle of $L_G[X]$ bounding $r$. Since $|V(\gamma)| \geq |S| \geq 4$, $r$ is non-empty 
(Proposition~\ref{prop:compo_in_region}) and hence contains
a component $C \in \calC_G(X)$. Therefore, $V(\gamma)$ must be a minimal separator of $G$, due to
Lemma~\ref{lem:BT-minsep-PMC}.
But this is impossible, since some edge of $L_G[S]$ is a chord of $\gamma$, contradicting
Proposition~\ref{prop:min-sep-of-triconn}.
\end{proof}

The following is our main theorem in this section.
\begin{theorem}
\label{thm:PMC-is-steering}
    A vertex set $X$ of  a triconnected plane graph $G$ is a PMC if and only if
    $L_G[X]$ is a steering.
\end{theorem}
\begin{proof}
Lemma~\ref{lem:steering-is-PMC} shows that if $L_G[X]$ is a steering then $X$ is a PMC.
We prove the other direction. Let $X$ be a PMC of $G$.

If $|N_G(C)| = 3$ for
every $C \in \calC_G(X)$, $L_G[X]$ is a $K_4$ due to Lemma~\ref{lem:all-triang} and 
hence is an $(X \setminus \{x\}, \{x\})$-steering for every $x \in X$.

Suppose there is some $C \in \calC_G(X)$ with $|S| \geq 4$ where $S = N_G(C)$.
Due to Lemma~\ref{lem:pmc-has-arch}, there is an arch of $S$ that is a subset of
$X \setminus S$. Let $P$ be an inclusion-wise minimal arch of $S$ that is a subset of $X \setminus S$.
Due to Lemma~\ref{lem:min-arch2steering}, $L_G[S \cup P]$ is a steering and hence $S \cup P$
is a PMC of $G$, due to Lemma~\ref{lem:steering-is-PMC}. 
Due to Proposition~\ref{prop:subset-enough}, $X = S \cup P$ and hence
$L_G[X]$ is a steering.
\end{proof}

\section{Polynomial delay generation of PMCs of a triconnected planar graph}
\label{sec:poly_delay}
Let $J$ be some combinatorial structure of size $n$ and let $\calS(J)$
be the set of some combinatorial objects defined on $J$. We say an algorithm
\emph{generates $\calS(J)$ with polynomial delay} if it outputs each member
of $\calS(J)$ exactly once and the time between two consecutive events is
$n^{O(1)}$, where an event is the initiation of the algorithm, an output,
or the termination of the algorithm.

The basic approach for generating $\Pi(G)$ for a given triconnected plane graph $G$
is as follows.
PMCs $X$ such that $|N_G(C)| = 3$ for every $C \in \calC_G(X)$ can be trivially
generated with polynomial delay, since $L_G[X]$ is a $K_4$ for each such $X$
due to Lemma~\ref{lem:all-triang}.
So, we concentrate on generating $\Pi'(G)$ consisting of members $X$ of $\Pi(G)$
such that $\calC_G(X)$ contains a component $C$ with $|S| \geq 4$ where $S = N_G(C)$.

We need an algorithm to generate minimal separators of $G$.
Although a polynomial delay algorithm for this task is known for general graphs~\cite{berry2000generating}, we use an algorithm specialized for triconnected plane graphs for a technical reason to become clear below.

We need the following result due to Uno and Satoh \cite{uno2014efficient}. 
\begin{lemma}
[Uno and Satoh \cite{uno2014efficient}]
 \label{lem:chordless}
    Given a graph $G$, all chordless cycles of $G$ can be generated
    with polynomial delay. 
    Given a graph $G$ and two vertices $s, t \in V(G)$, 
    all chordless paths of $G$ between $s$ and $t$ can be
    generated with polynomial delay.  
\end{lemma}

\begin{lemma}
    \label{lem:minsep-polydelay}
    Given a triconnected plane graph $G$, all minimal separators of $G$
    can be generated with polynomial delay. Moreover, for each $v \in V(G)$,
    all minimal separators of $G$ that do not contain $v$ can be generated with polynomial delay.
\end{lemma}
\begin{proof}
    Given $G$, we generate all chordless cycles of $L_G$ with polynomial delay,
    using the algorithm of Uno and Satoh. These chordless cycles are in one-to-one correspondence
    with minimal separators of $G$
    due to Proposition~\ref{prop:min-sep-of-triconn}.
    For the second part, we generate all chordless cycles of $G[V(G) \setminus \{v\}]$.
    Since a chordless cycle of $G[V(G) \setminus \{v\}]$ is a chordless cycle of $G$, we have the result.
\end{proof}
\begin{remark}
    It is not clear if the algorithm \cite{berry2000generating} for generating minimal separators of a general graph can be adapted to support the second part of the lemma, which is used in the proof of Theorem~\ref{thm:poly-delay-pmcs}. We also note that the chordless path part of Lemma~\ref{lem:chordless} is used in
    the proof of Lemma~\ref{lem:gen-for-compo}.
\end{remark}

To design an algorithm to generate all PMCs based on our characterization of PMCs stated in Theorem~\ref{thm:PMC-is-steering},
we need some preparations.
\begin{definition}
    Let $G$ be a triconnected plane graph and $C$ a minimally separated component of $G$
    with $|S| \geq 4$ where $S = N_G(C)$.
    Recall that there are exactly two full components associated with $S$ and 
    let $C'$ be the full component distinct from $C$.
    A \emph{port} of $(S, C)$ is a vertex $u \in C'$ such that
    $N_{L_G}(u) \cap S$ is a slot of $L_G[S]$. We call the slot $N_{L_G}(u) \cap S$
    of $L_G[S]$ the \emph{slot for} $u$. A pair of ports $u_1$ and $u_2$ is \emph{valid}
    if the union of the slot for $u_1$ and the slot for $u_2$ is not a slot.
    A vertex $s \in S$ with two neighbors $s_1$ and $s_2$ on the cycle $L_G[S]$ 
    is a \emph{hinge} of a valid pair of ports $u_1$ and $u_2$
    if the slot for $u_1$ is either $\{s_1\}$ or $\{s_1, s\}$ and
    the slot for $u_2$ is either $\{s_2\}$ or $\{s_2, s\}$.
    See Figure~\ref{fig:ports} for examples of ports, valid and invalid port pairs, and
    a hinge.
\end{definition}
\begin{figure}[hbtp]
\begin{center}
\includegraphics[width=12cm]{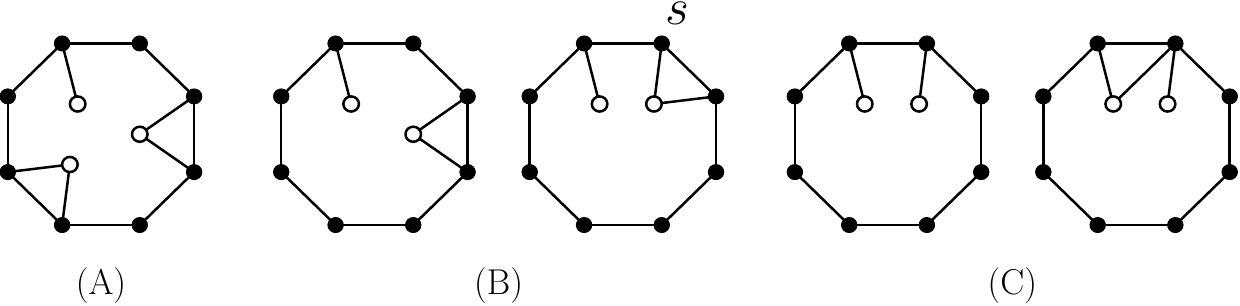}
\caption{(A) White vertices are some ports of $(S, C)$ where $S$ consists of black vertices and
$C$ is the full component associated with $S$ that lies in the outer face of the black cycle.
(B) Some valid pairs of ports. The valid pair in the second example has a hinge $s$. (C) Some invalid pairs of ports.}
\label{fig:ports}
\end{center}
\end{figure}

Ports are potential ends of a path $p$ such that $L_G[X]$ for $X = S \cup V(p)$ is a steering.
A pair of ports must be valid in order for the pair to be the ends of such path $p$.
We generate such paths using the algorithm for generating chordless paths.
The following lemma is the basis of this approach.
\begin{lemma}
    \label{lem:steerings-gen}   
Let $G$ be a triconnected plane graph and $C$ a minimally separated component of $G$
with $|S| \geq 4$, where $S = N_G(C)$. Let $C'$ be the other full component of $S$.
Let $P$ be an arbitrary subset of  $C'$ with $|P| \geq 2$ 
and let $X = S \cup P$.
\begin{enumerate}
    \item For each valid pair $u_1$ and $u_2$ of ports of $(S, C)$, let 
$A(C, u_1, u_2)$ be the subgraph of $L_G$ induced by
the vertex set $(C' \setminus N_{L_G}(S)) \cup \{u_1, u_2\}$.
Then, $L_G[X]$ is an $(S, P)$-steering if and only if
there is a valid pair $u_1$ and $u_2$ of the ports of $S$ such that 
$P = V(p)$ for some chordless path $p$ of $A(C, u_1, u_2)$ between $u_1$ and $u_2$.
\item For each valid pair $u_1$ and $u_2$ of ports of $(S, C)$ that has hinges and
each hinge $s$ of the pair, let $B(C, u_1, u_2, s)$ be the subgraph of $L_G$ induced by
the vertex set $(C' \setminus N_{L_G}(S)) \cup N_{L_G}(s) \cup \{u_1, u_2\}$.
Then, $L_G[X]$ is an $(S', \{s\})$-steering for $s \in S$, where $S' = X \setminus \{s\}$,
if and only if there is a valid pair of ports $u_1$ and $u_2$ of $(S, C)$ with a hinge $s$
such that $P = V(p)$ for some chordless path $p$ of 
$B(C, u_1, u_2, s)$ between $u_1$ and $u_2$.
\end{enumerate}
\end{lemma}
\begin{proof}
For part 1, 
suppose $L_G[X]$ is an $(S, P)$-steering where $P \subseteq C'$.
Let $H = L_G[X]$.
Due to Definition~\ref{def:steering}, $H[P] = L_G[P]$ is a path.
Let $u_1$ and $u_2$ be the two ends of this path. Then,
$u_1$ and $u_2$ are ports of $S$ and, moreover, the pair of these ports is valid.
Since no internal vertex of the path $L_G[P]$ is adjacent in $L_G$ to any vertex in $S$,
$L_G[P]$ is a path in $A(C, u_1, u_2)$. 
It is a chordless path, since if it had a chord in $A(C, u_1, u_2)$, then the graph $L_G[P]$ would contain
that chord and $L_G[P]$ would not be a path.

For the other direction, suppose that there is a valid pair of ports
$u_1$ and $u_2$ of $S$ such that $L_G[P]$ is a chordless path of $A(C, u_1, u_2)$ between $u_1$ and $u_2$.
Then, $N_{L_G}(P \setminus \{u_1, u_2\}) \cap S$ is empty and hence 
$L_G[S \cup P]$ is an $(S, P)$-steering.

For part 2, 
suppose $L_G[X]$ is an $(S', \{s\})$-steering where $s \in S$ and $S' = X \setminus \{s\}$.
Due to the proof of Lemma~\ref{lem:min-arch2steering}, $L_G[P]$ is a path
of $L_G$ between some valid pair of ports $u_1$ and $u_2$ of $(S, C)$ and, moreover,
$s$ is a hinge of this pair of ports.
Since the only vertex in $S$ that can be adjacent to any internal vertex of $L_G[P]$ is $s$,
$L_G[P]$ is a path in $B(C, u_1, u_2, s)$ between $u_1$ and $u_2$.
It is a chordless path, since if it had a chord in $B(C, u_1, u_2, s)$, then the graph $L_G[P]$ would contain
that chord and $L_G[P]$ would not be a path.

For the other direction, suppose that there is a valid pair of ports
$u_1$ and $u_2$ of $S$ with a hinge $s$ such that $L_G[P]$ is a chordless path of $B(C, u_1, u_2, s)$ between $u_1$ and $u_2$.
Then, $L_G[S']$, where $S' = X \setminus \{s\}$, is a cycle and 
$L_G(X)$ is an $(S', \{s\})$-steering.
\end{proof}

In implementing the generation of PMCs based on this lemma, 
we need to deal with the problem of suppressing duplicate
outputs of a single element without introducing super-polynomial delay.
We use the technique due to 
Bergougnoux, Kant\'e and Wasa \cite{bergougnoux2019disjunctive} to address this problem.
They used this technique for a particular problem of generating what they call minimal disjunctive separators.
We formulate their technique in the following theorem that can be used for wide range of generation problems.

Let $G$ be a graph on $n$ vertices.
We consider the general task of generating some structures defined on $G$.
Let $\calS(G)$ denote the set of those structures to be generated. 
Suppose $N$ subsets $\calS_1(G)$, \ldots, $\calS_N(G)$ of $\calS(G)$ are defined.
Let $I = \{i \mid 1 \leq i \leq N\}$. 
Suppose that the following conditions are satisfied.
\begin{enumerate}
    \item \label{item:union}
    $\calS(G) = \bigcup_{i \in I} \calS_i(G)$.
    \item $N = n^{O(1)}$.
    \item For each $s \in \calS(G)$ and $i \in I$, it can be decided whether $s \in \calS_i(G)$ in
        time $n^{O(1)}$.
    \item For each $i \in I$, there is an algorithm $\gen_i$ 
    that generates $\calS_i(G)$ with polynomial delay.
\end{enumerate}

\begin{theorem} \label{thm:poly-delay}
    Under the above assumptions, $\calS(G)$ can be generated with polynomial delay.
\end{theorem}
We prove this theorem by showing that Algorithm~\ref{alg:poly-delay-gen} generates $\calS(G)$ with
polynomial delay. The idea of this algorithm is as follows.
For each $s \in \calS(G)$, let us say that $i \in I$ \emph{owns} $s$ if $i$ is the largest
$j$ such that $s \in \calS_j(G)$.
We let $\gen_i$, where $i$ owns $s$, be responsible for the output of $s$ and suppress the output of $s$
from $\gen_j$ for other $j$. The executions of $\gen_i$ for
$i \in I$ are scheduled in such a way that at most a polynomial number of outputs are suppressed before
an unsuppressed output occurs.
\begin{algorithm}[t]
\caption{Generation of $\calS(G)$, using sub-generators $\gen_i$ for $\calS_i(G)$, $i \in I$}
\Comment{We say \emph{emit} $s$ to mean that this algorithm outputs $s$, in order to distinguish this action from the output events of the sub-generators}\label{alg:poly-delay-gen}
\LineComment{We say that an output of a sub-generator is \emph{suppressed} if it is not emitted}
\begin{algorithmic}[1]
\Ensure Generate all members of $\calS(G)$
\For {each $i \in I$}
 \State Initiate $\gen_i$ and immediately suspend it
\EndFor
\State $i^* \gets 1$
\Repeat
 \State $\mathop{ascending} \gets \mathop{true}$
 \While {$\mathop{ascending}$}
   \State Execute $\gen_{i^*}$ up to the next event
       \label{line:while_top}
   \If {the event is the termination event}
     \State $\mathop{ascending} \gets \mathop{false}$
   \Else \Comment{the event is an output event}
     \State Let $s$ be the object to be output
       \If {$s \in \calS_i(G)$ for some $i > i^*$}
          \Comment{the output of $s$ is suppressed}
       \State $i^* \gets $ the smallest $j > i^*$
         such that $\gen_j$ has not been terminated
         \label{line:suppress}
         \State 
         \Comment{such $j$ exists, see the proof}
       \Else   
       \State Emit $s$ 
       \State $\mathop{ascending} \gets \mathop{false}$
       \EndIf
   \EndIf
 \EndWhile 
 \If {there is some $j$ such that $\gen_j$ has not been terminated}
   \State $i^* \gets$ the smallest $j$ such that $\gen_j$ has not been terminated
    \label{line:down}
 \EndIf
 \Until {$\gen_i$ has been terminated for every $i \in I$}
\end{algorithmic}
\end{algorithm}

\begin{lemma}
    \label{lem:poly-delay}
    \cref{alg:poly-delay-gen} generates $\calS(G)$ with polynomial delay,
    if all of the four conditions above are satisfied.
\end{lemma}
\begin{proof}
Let us say that $i \in I$ \emph{owns} $s \in \calS(G)$ if $s \in \calS_i(G)$ but
$s \not\in \calS_{i'}(G)$ for every $i' > i$. We also say that $\gen_i$ is the owner of
if $i$ owns $s$. Note that the set of
members of $\calS(G)$ owned by $i$ is the set of those generated by $\gen_i$ without
being suppressed in Algorithm~\ref{alg:poly-delay-gen}.

We first prove the claim given as a comment in the algorithm:
when an output of $\gen_i$ is suppressed at Line~\ref{line:suppress}, there
is some $j > i$ such that $\gen_{j}$ has not been terminated.
Let $\calS_i'(G)$ denote the set of members of $\calS_i(G)$ that are not owned by $i$.
At each point in the algorithm execution, let $\epsilon_i$ denote the
number of elements of $\calS_i'(G)$   that have been emitted and $\sigma_i$ the
number of elements of $\calS_i'(G)$ whose output by $\gen_i$ has been suppressed.
We claim that the invariant $\epsilon_i \leq \sigma_i$
holds at line~\ref{line:while_top}, the beginning of the while loop, for every $i \in I$. 
This certainly holds at the start of the algorithm. Before a member $s$ of $\calS_i'(G)$ 
is emitted for the first time, a suppression of the output of some member $s'$ of $\calS_i'(G)$ 
by $\gen_i$ must occur, in order for the owner of $s$ to obtain the control.
Unless $\sigma_i = |\calS_i'(G)|$, $\gen_i$ has not been terminated and the control
transfers to $\sigma_j$ for some $j \leq i$ after the emission of $s$.
Thus, every increment of $\epsilon_i$ is preceded by an increment of $\sigma_i$
and therefore the invariant holds. Therefore, at line~\ref{line:suppress},
$\sigma_{i^*} > \epsilon_{i^*}$ holds and the owner of some $s \in \calS_{i^*}'(G)$ has not been
terminated. This proves the claim maid in the comment at at line~\ref{line:suppress}.

Thus, every step of the algorithm is well-defined and hence it terminates after executing
$\gen_i$ for every $i \in I$ to termination. Due to Condition~\ref{item:union},
every member of $\calS(G)$ is owned by some $i \in I$ and hence is emitted exactly once, 
when it is output by its owner.

The delay in this algorithm is polynomial time, 
since each while iteration runs in polynomial time, 
each while iteration ends with an emission or the termination event of a sub-generator, and 
there are only $N = n^{O(1)}$ termination events. 
\end{proof}

We now turn to the application of this technique to our goal.

For a minimally separated component $C$ of a triconnected plane graph $G$ with $|S| \geq 4$ 
where $S = N_G(C)$,
let $\Pi(G, C)$ denote the set of PMCs $X$ such that $C \in \calC_G(X)$ and
$L_G[X]$ is a steering.

\begin{proposition}
\label{prop:pi-g-c-nonempty}
    Let $G$ be a triconnected plane graph and $C$ a minimally separated component of $G$
    such that $|N_G(C)| \geq 4$. Then, $\Pi(G, C) \neq \emptyset$.
\end{proposition}
\begin{proof}
Let $S = N_G(C)$ and let $C'$ be the other full component of $S$.
Let $H$ be an arbitrary minimal triangulation of $G$ in which $S$ is a clique.
Let $X$ be an arbitrary maximal clique of $H[S \cup C']$ that contains $S$.
Since $S$ is a separator of $H$, $X$ is a maximal clique of $H$.
Since no maximal clique of a chordal graph is a minimal separator, $X$ is a proper
superset of $S$ and belongs to $\Pi(G, C)$.
\end{proof}

\begin{lemma}
    \label{lem:gen-for-compo}
    Given a triconnected plane graph $G$ and a minimally separated component $C$ of $G$ with 
    $|S| \geq 4$ where $S = N_G(C)$, 
    $\Pi(G, C)$ can be generated with polynomial delay.
\end{lemma}
\begin{proof}
Based on Lemma~\ref{lem:steerings-gen}, we generate $\Pi(G, C)$ as follows.
We first generate $X \in \Pi(G, C)$ such that $L_G[X]$ is an $(S, P)$-steering where $P = X \setminus S$.
To do so, for every valid pair of ports $u_1$ and $u_2$ of
$S$, we generate all chordless paths of $A(C, u_1, u_2)$ between $u_1$ and $u_2$,
where $A(C, u_1, u_2)$ is as defined in Lemma~\ref{lem:steerings-gen}, using the algorithm due to Sato and Uno
given in Lemma~\ref{lem:chordless}. For each path $p$ generated, we output $S \cup V(p)$.
Since the number of valid pairs of ports is $O(n^2)$, the generation of $X$ in this category is with polynomial delay
even if some valid pairs do not yield any such $X$.

We then generate $X \in \Pi(G, C)$ such that $L_G[X]$ is an $(S', \{s\})$-steering where $s \in S$ and
$S' = X \setminus \{s\}$.
To do so, for every valid pair of ports $u_1$ and $u_2$ of
$S$ with hinges and for each hinge $s$ of this pair, we generate all chordless paths of 
$B(C, u_1, u_2, s)$ between $u_1$ and $u_2$. For each path $p$ generated, we output $S \cup V(p)$.
The valid pair of $u_1$ and $u_2$ may have more than one ports when 
$|S| = 4$, which could cause duplicate outputs of an identical $X$. We suppress such duplicate outputs
using Theorem~\ref{thm:poly-delay}. Let $\calT$ be the set of all triples $(u_1, u_2, s)$
where $(u_1, u_2)$ is a valid pair of ports of $(S, C)$ and $s$ is a hinge of this pair.
Index the set $\calT$ by $I = \{1, \ldots, N\}$ where $N = |\calT|$.
For each $i \in I$, let $\calP_i$ denote the
set of chordless paths of $B(C, u_1, u_2, s)$ between $u_1$ and $u_2$, where
$(u_1, u_2, s)$ is the triple in $\calT$ indexed by $i$. Let $\calP = \bigcup_{i \in I} \calP_i$.
Our goal is to generate $\calP$ without duplication. We verify the conditions for Theorem~\ref{thm:poly-delay}
to be applied.
\begin{enumerate}
\item $\calP = \bigcup_{i \in I} \calP_i$ by definition.
\item $N = |I| = O(n^3)$.
\item For each path $p$ in $\calP$ and each $i \in I$, it can be decided whether $p \in \calP_i$ in
        time $n^{O(1)}$.
\item For each $i \in I$, we can generate $\calP_i$ with polynomial delay as described above.
\end{enumerate}
Therefore, Theorem~\ref{thm:poly-delay} applies and we can generate $\calP$ and hence all PMCs $X$ 
in this category, with polynomial delay.
\end{proof}
\begin{remark}
Since the PMCs $X$ in the second category in the above lemma, those such that
$L_G[X]$ is an $(S', \{s\})$-steering for some $s \in S$ where $S' = X \setminus \{s\}$,
also belong to $\Pi(G, C')$ for some full component associated with $S'$,
it might appear that the cumbersome generation of those PMCs described in
the above proof is unnecessary. However, if we generate only the members of $\Pi(G, C)$ in the
first category, those $X$ such that $L_G[X]$ is an $(S, P)$-steering where $P = X \setminus S$,
then it is possible for the generator of $\Pi(G, C)$ to produce no outputs.
This is a problem, since there is no readily provable polynomial bound on the number of 
successive components $C$ for which the generation of $\Pi(G, C)$ produces no outputs.
\end{remark}

The following is the main result of this section.
\begin{theorem}
\label{thm:poly-delay-pmcs}
Given a triconnected plane graph $G$, 
$\Pi(G)$ can be generated with polynomial delay.
\end{theorem}
\begin{proof}
It suffices to show that $\Pi'(G)$, the set of PMCs $X$ of $G$ such that $|N_G(C)| \geq 4$
for some $C \in \calC_G(X)$,
can be generated with polynomial delay. Let $n = |V(G)|$.
    For each $v \in G$, let $\Pi_v(G)$ denote the
    union of $\Pi(G, C)$ over all minimally separated components $C$ of
    $G$ such that $v \in C$ and $|N_G(C)| \geq 4$.
    We have $\Pi'(G) = \bigcup_{v \in V(G)} \Pi_v(G)$ and $|V(G)| = n = n^{O(1)}$.
    Given $X \in \Pi(G)$ and $v \in V(G)$, it is straightforward to test if $X \in \Pi_v(G)$
    in polynomial time. For each $v$, $\Pi_v(G)$ can be
    generated with polynomial delay: we generate all minimal separators 
    $S$ such that $v \not\in S$ by the second part of Lemma~\ref{lem:minsep-polydelay}
    and, for each $S$ generated such that $|S| \geq 4$, generate
    $\Pi(G, C)$ for the full component $C$ of $S$ such that $v \in C$. 
    The delay in this algorithm is polynomial, since minimal separators are generated with polynomial delay and,
    for each minimal separator generated, at least one PMC is generated due to
    Proposition~\ref{prop:pi-g-c-nonempty}.
    
    Therefore, applying Theorem~\ref{thm:poly-delay}, we can generate $\Pi(G)$
    with polynomial delay.
\end{proof}

\section{Computing the planar treewidth}
An algorithm for generating PMCs immediately leads to a treewidth algorithm for triconnected planar graphs
because of Theorem~\ref{thm:bt-dp} due to Bouchitt{\'e} and Todinca~\cite{bouchitte2001treewidth}.
To extend the algorithm for general planar graphs, we use the following result from \cite{bodlaender2006safe}.
A vertex set $S$ of a graph $G$ is an \emph{almost clique} of $G$ if there is some $v \in S$ such that
$S \setminus \{v\}$ is a clique of $G$.
\begin{theorem}[Bodlaender and Koster \cite{bodlaender2006safe}]\label{thm:almost_clique}
Let $G$ be a graph and $S$ a minimal separator of $G$ that is an almost clique.
Let $C_i$, $1 \leq i \leq m$, be the components of $G[V(G) \setminus S]$ and let
$G_i = G[C_i \cup N_G(C_i)] \cup K(N_G(C_i))$, for $i = 1, \ldots, m$, be the
graph obtained from $G[C_i \cup N_G(C_i)]$, the subgraph of $G$ induced by the closed neighborhood of $C_i$,
by filling the open neighborhood of $C_i$ into a clique. Then $\tw(G) = \max\{|S|, \max_{i \in \{1, \ldots, m\}} \tw(G_i)\}$.
\end{theorem}

We use this theorem for the special case where $S$ is a two-vertex minimal separator: such $S$ is always an almost clique.

\begin{theorem}
Given a planar graph $G$, $\tw(G)$ can be computed in time $|\Pi(G)| n^{O(1)}$.
\end{theorem}
\begin{proof}
We assume that $G$ is biconnected. Otherwise, we compute the treewidth of each biconnected component and
take the maximum.

We first prove the case where $G$ is triconnected. We compute a planar embedding of $G$ in time $O(n)$ \cite{HopcroftT74}.
Then we apply our generation algorithm to compute $\Pi(G)$ in time $|\Pi(G)| n^{O(1)}$.
Finally, we apply Theorem~\ref{thm:bt-dp} to obtain $\tw(G)$ in time $|\Pi(G)| n^{O(1)}$.

For a general planar graph $G$, let $s(G)$ denote the number of two-vertex separators of $G$.
We prove the theorem by induction on $s(G)$. The base case $s(G) = 0$, where $G$ is triconnected, is already proved.
Suppose $s(G) > 0$ and let $S$ be an arbitrary two-vertex separator of $G$.
Let $C_1$, \ldots, $C_m$ be
the components of $G[V(G) \setminus S]$. Since $G$ is biconnected, we have $N(C_i) = S$ for every $i \in \{1, \dots, m\}$.
Therefore, $S$ is a minimal separator of $G$ and, due to Theorem~\ref{thm:almost_clique},
we have $\tw(G) = \max\{2, \max_{i \in \{1, \ldots, m\}} \tw(G_i)\}$ where $G_i = G[C_i \cup S] \cup K(S)$.
Since $G_i$ is a subgraph of a minor of $G$ obtained by contracting $C_{i'} \cup {v}$ into a single vertex 
for arbitrary $i' \neq i$ and
$v \in S$, $G_i$ is planar. Moreover, $G_i$ is biconnected since any cut vertex of $G_i$ would be a cut vertex of $G$.
Therefore, we may apply the induction hypothesis to each $G_i$ and 
compute $\tw(G_i)$ for each $i$ in time $|\Pi(G_i)| n^{O(1)}$.
We claim that $\sum_i |\Pi(G_i)| \leq |\Pi(G)|$. To see this, let $X$ be an arbitrary PMC of $G_i$ and
$H$ a minimal triangulation of $G_i$ in which $X$ is a maximal clique. Let $H'$ be an arbitrary minimal triangulation
of $G$ such that $H'[V(G_i)] = H$. Since every minimal separator of $G$ that is a clique in a minimal triangulation of
$G$ is a minimal separator of that triangulation \cite{heggernes2006minimal}, $S$ is a minimal separator of $H'$.
Therefore, $X$ is a maximal clique of $H'$ and hence is a PMC of $G$. Moreover, 
$X$ cannot be a PMC of $G_j$ for any $j \neq i$ since $X$ is not a subset of $V(G_j)$. Therefore, the claim holds.
We conclude that 
we can compute $\tw(G)$ in time $|\Pi(G)| n^{O(1)}$.
\end{proof}

\section{Future work}
Our success in a special case of the open question whether $\Pi(G)$ can be computed in time $|\Pi(G)|n^{O(1)}$
does not seem to suggest any approach to solving the question for general graphs.
Much more modest goal is to address the question for general planar graphs. 
It might be possible to extend our generation algorithm to general planar graphs by finding some way of
handling PMCs that cross two-vertex minimal separators. Bart Jansen asked if the generation result could be extended to
graphs embedded on tori. We observe that latching graphs are not appropriate tool for that purpose, since
if $G$ is a triconnected graph embedded on a torus and $X$ is a PMC of $G$
then the subgraph of the latching graph of $G$, appropriately defined on tori,
induced by $X$ may have edge crossings: the proof of Proposition~\ref{prop:pmc-planar} relies essentially on the assumption
that $G$ is embedded in the sphere.
Studying upper bounds on $|\Pi(G)|$ is another avenue of research. There is a class of planar graphs 
for which $|\Pi(G)| = \Omega(1.442^n)$ \cite{FominKTV08}. These graphs are biconnected but not triconnected, so it 
would be interesting to ask if an upper bound smaller than this bound holds for triconnected planar graphs.

\bibliography{references}

\begin{thebibliography}{10}

\bibitem{bergougnoux2019disjunctive}
Benjamin Bergougnoux, Mamadou~Moustapha Kant\'e, and Kunihiro Wasa.
\newblock Disjunctive minimal separators enumeration.
\newblock In {\em International Workshop on Enumeration Problems and Applications}, 2019.
\newblock \url{https://benjaminbergougnoux.github.io/pdf/bkw19.pdf}.

\bibitem{berry2000generating}
Anne Berry, Jean~Paul Bordat, and Olivier Cogis.
\newblock Generating all the minimal separators of a graph.
\newblock {\em Int. J. Found. Comput. Sci.}, 11(3):397--403, 2000.
\newblock \href {https://doi.org/10.1142/S0129054100000211} {\path{doi:10.1142/S0129054100000211}}.

\bibitem{10.1137/S0097539793251219}
Hans~L. Bodlaender.
\newblock A linear-time algorithm for finding tree-decompositions of small treewidth.
\newblock {\em SIAM J. Comput.}, 25(6):1305–1317, December 1996.
\newblock \href {https://doi.org/10.1137/S0097539793251219} {\path{doi:10.1137/S0097539793251219}}.

\bibitem{bodlaender2006open}
Hans~L Bodlaender, Leizhen Cai, Jianer Chen, Michael~R Fellows, Jan~Arne Telle, and D{\'a}niel Marx.
\newblock Open problems in parameterized and exact computation—{IWPEC 2006}.
\newblock {\em Department of Information and Computing Sciences, Utrecht University}, 2006.

\bibitem{bodlaender2006safe}
Hans~L. Bodlaender and Arie M. C.~A. Koster.
\newblock Safe separators for treewidth.
\newblock {\em Discret. Math.}, 306(3):337--350, 2006.
\newblock \href {https://doi.org/10.1016/J.DISC.2005.12.017} {\path{doi:10.1016/J.DISC.2005.12.017}}.

\bibitem{bouchitte2003chordal}
Vincent Bouchitt{\'{e}}, Fr{\'{e}}d{\'{e}}ric Mazoit, and Ioan Todinca.
\newblock Chordal embeddings of planar graphs.
\newblock {\em Discret. Math.}, 273(1-3):85--102, 2003.
\newblock \href {https://doi.org/10.1016/S0012-365X(03)00230-9} {\path{doi:10.1016/S0012-365X(03)00230-9}}.

\bibitem{bouchitte2001treewidth}
Vincent Bouchitt{\'{e}} and Ioan Todinca.
\newblock Treewidth and minimum fill-in: Grouping the minimal separators.
\newblock {\em {SIAM} J. Comput.}, 31(1):212--232, 2001.
\newblock \href {https://doi.org/10.1137/S0097539799359683} {\path{doi:10.1137/S0097539799359683}}.

\bibitem{bouchitte2002listing}
Vincent Bouchitt{\'{e}} and Ioan Todinca.
\newblock Listing all potential maximal cliques of a graph.
\newblock {\em Theor. Comput. Sci.}, 276(1-2):17--32, 2002.
\newblock \href {https://doi.org/10.1016/S0304-3975(01)00007-X} {\path{doi:10.1016/S0304-3975(01)00007-X}}.

\bibitem{FominKTV08}
Fedor~V. Fomin, Dieter Kratsch, Ioan Todinca, and Yngve Villanger.
\newblock Exact algorithms for treewidth and minimum fill-in.
\newblock {\em {SIAM} J. Comput.}, 38(3):1058--1079, 2008.
\newblock \href {https://doi.org/10.1137/050643350} {\path{doi:10.1137/050643350}}.

\bibitem{fomin2011exact}
Fedor~V. Fomin, Ioan Todinca, and Yngve Villanger.
\newblock Exact algorithm for the maximum induced planar subgraph problem.
\newblock In Camil Demetrescu and Magn{\'{u}}s~M. Halld{\'{o}}rsson, editors, {\em Proceedings of the 19th Annual European Symposium, {ESA} 2011}, volume 6942 of {\em Lecture Notes in Computer Science}, pages 287--298. Springer, 2011.
\newblock \href {https://doi.org/10.1007/978-3-642-23719-5_25} {\path{doi:10.1007/978-3-642-23719-5_25}}.

\bibitem{FominTV15}
Fedor~V. Fomin, Ioan Todinca, and Yngve Villanger.
\newblock Large induced subgraphs via triangulations and {CMSO}.
\newblock {\em {SIAM} J. Comput.}, 44(1):54--87, 2015.
\newblock \href {https://doi.org/10.1137/140964801} {\path{doi:10.1137/140964801}}.

\bibitem{FominV10}
Fedor~V. Fomin and Yngve Villanger.
\newblock Finding induced subgraphs via minimal triangulations.
\newblock In Jean{-}Yves Marion and Thomas Schwentick, editors, {\em 27th International Symposium on Theoretical Aspects of Computer Science, {STACS} 2010, March 4-6, 2010, Nancy, France}, volume~5 of {\em LIPIcs}, pages 383--394. Schloss Dagstuhl - Leibniz-Zentrum f{\"{u}}r Informatik, 2010.
\newblock \href {https://doi.org/10.4230/LIPICS.STACS.2010.2470} {\path{doi:10.4230/LIPICS.STACS.2010.2470}}.

\bibitem{fomin2012treewidth}
Fedor~V. Fomin and Yngve Villanger.
\newblock Treewidth computation and extremal combinatorics.
\newblock {\em Comb.}, 32(3):289--308, 2012.
\newblock \href {https://doi.org/10.1007/S00493-012-2536-Z} {\path{doi:10.1007/S00493-012-2536-Z}}.

\bibitem{heggernes2006minimal}
Pinar Heggernes.
\newblock Minimal triangulations of graphs: A survey.
\newblock {\em Discrete Mathematics}, 306(3):297--317, 2006.

\bibitem{HopcroftT74}
John~E. Hopcroft and Robert~Endre Tarjan.
\newblock Efficient planarity testing.
\newblock {\em J. {ACM}}, 21(4):549--568, 1974.
\newblock \href {https://doi.org/10.1145/321850.321852} {\path{doi:10.1145/321850.321852}}.

\bibitem{inkmann2008tree}
Torsten Inkmann.
\newblock {\em Tree-based decompositions of graphs on surfaces and applications to the Traveling Salesman Problem}.
\newblock PhD thesis, Georgia Institute of Technology, 2008.
\newblock URL: \url{http://hdl.handle.net/1853/22583}.

\bibitem{KAMMER201660}
Frank Kammer and Torsten Tholey.
\newblock Approximate tree decompositions of planar graphs in linear time.
\newblock {\em Theoretical Computer Science}, 645:60--90, 2016.
\newblock \href {https://doi.org/10.1016/j.tcs.2016.06.040} {\path{doi:10.1016/j.tcs.2016.06.040}}.

\bibitem{korhonen2023improved}
Tuukka Korhonen and Daniel Lokshtanov.
\newblock An improved parameterized algorithm for treewidth.
\newblock In Barna Saha and Rocco~A. Servedio, editors, {\em Proceedings of the 55th Annual {ACM} Symposium on Theory of Computing, {STOC} 2023}, pages 528--541. {ACM}, 2023.
\newblock \href {https://doi.org/10.1145/3564246.3585245} {\path{doi:10.1145/3564246.3585245}}.

\bibitem{seymour1994call}
Paul~D. Seymour and Robin Thomas.
\newblock Call routing and the ratcatcher.
\newblock {\em Comb.}, 14(2):217--241, 1994.
\newblock \href {https://doi.org/10.1007/BF01215352} {\path{doi:10.1007/BF01215352}}.

\bibitem{tamaki2019positive}
Hisao Tamaki.
\newblock Positive-instance driven dynamic programming for treewidth.
\newblock {\em J. Comb. Optim.}, 37(4):1283--1311, 2019.
\newblock \href {https://doi.org/10.1007/S10878-018-0353-Z} {\path{doi:10.1007/S10878-018-0353-Z}}.

\bibitem{uno2014efficient}
Takeaki Uno and Hiroko Satoh.
\newblock An efficient algorithm for enumerating chordless cycles and chordless paths.
\newblock In Saso Dzeroski, Pance Panov, Dragi Kocev, and Ljupco Todorovski, editors, {\em Proceedings of the 17th International Conference of Discovery Science, {DS} 2014}, volume 8777 of {\em Lecture Notes in Computer Science}, pages 313--324. Springer, 2014.
\newblock \href {https://doi.org/10.1007/978-3-319-11812-3_27} {\path{doi:10.1007/978-3-319-11812-3_27}}.

\bibitem{Whitney1933}
Hassler Whitney.
\newblock 2-isomorphic graphs.
\newblock {\em American Journal of Mathematics}, 55(1):245--254, 1933.
\newblock \href {https://doi.org/10.2307/2371127} {\path{doi:10.2307/2371127}}.

\end{thebibliography}

\end{document}